\documentclass[12pt]{article}
  
\usepackage{cite}
\usepackage{appendix}
\usepackage[bookmarks,colorlinks,breaklinks]{hyperref}  
\hypersetup{linkcolor=blue,citecolor=blue,filecolor=blue,urlcolor=blue} 
\usepackage{amsmath,amsfonts,amsthm,amssymb,xspace,bm,fullpage}

\newcommand{\newref}[2][]{\hyperref[#2]{#1~\ref*{#2}}}
\renewcommand{\eqref}[1]{\hyperref[#1]{(\ref*{#1})}}
\numberwithin{equation}{section}

\newcommand{\sref}[1]{\newref[Section]{#1}}

\newcommand{\tref}[1]{\newref[Theorem]{#1}}
\newcommand{\lref}[1]{\newref[Lemma]{#1}}

\newcommand{\cref}[1]{\newref[Corollary]{#1}}

\newcommand{\eref}[1]{\newref[Equation]{#1}}

\newcommand{\clref}[1]{\newref[Claim]{#1}}
\newcommand{\fcref}[1]{\newref[Fact]{#1}}
 
\theoremstyle{plain}
\newtheorem{theorem}{Theorem}[section]
\newtheorem{Thm}{Theorem}[section]

\newtheorem{lemma}[theorem]{Lemma}
\newtheorem{Lem}[theorem]{Lemma}

\newtheorem{claim}[theorem]{Claim}

\newtheorem{corollary}[theorem]{Corollary}

\newtheorem{Conj}[theorem]{Conjecture}

\newtheorem{definition}[theorem]{Definition}
\newtheorem{Def}[theorem]{Definition}

\newtheorem{fact}[theorem]{Fact}
\theoremstyle{definition}

\DeclareMathOperator*{\pr}{\mathsf{Pr}}
\DeclareMathOperator*{\E}{\mathsf{Pr}}

\DeclareMathOperator*{\ex}{\mathsf{E}}
\newcommand{\rgta}{\rightarrow}
\newcommand{\lfta}{\leftarrow}

\newcommand{\bias}{\mathsf{Bias}}
\newcommand{\cnf}{\ensuremath{\mathsf{CNF}}}
\newcommand{\dnf}{\ensuremath{\mathsf{DNF}}}

\newcommand{\ccnf}{\#\mathsf{\cnf}}
\newcommand{\cdnf}{\#\mathsf{DNF}}
\newcommand{\NP}{\ensuremath{\mathsf{NP}}}
\newcommand{\Ptime}{\ensuremath{\mathsf{P}}}

\newcommand{\eat}[1]{}
\newcommand{\etal}{{\em et al.}}

\newcommand{\hh}{\mathcal{H}}

\newcommand{\eps}{\epsilon}

\newcommand{\note}[1]{\marginpar{\tiny *note in TeX*}}
\newcommand{\ignore}[1]{}
\newcommand{\mc}[1]{\mathcal{#1}}

\newcommand{\calD}{{\cal D}}

\renewcommand{\phi}{\varphi}
\renewcommand{\epsilon}{\varepsilon}

\newcommand{\R}{\mathbb{R}}

\newcommand{\poly}{\mathrm{poly}}

\newcommand{\zo}{\{0,1\}}

\newcommand{\prg}{\ensuremath{\mathsf{PRG}}}
\newcommand{\dt}[1]{\ensuremath{\mathsf{DT}(#1)}}

\newcommand{\algo}[2]{
\bigskip
\fbox{
\parbox{6in}{
{ #1}\\
{\tt #2}
\bigskip
}}
\bigskip
}

\title{DNF Sparsification and a Faster Deterministic Counting Algorithm}

\author{Parikshit Gopalan\\
Microsoft Research\\
Silicon Valley\\
{\tt parik@microsoft.com}\\
\and
Raghu Meka\thanks{Work done while an intern at Microsft Research,
  Silicon Valley.}\\
Institute for Advanced Study\\
Princeton\\
{\tt raghu@ias.edu}\\
\and
Omer Reingold\\
Microsoft Research\\
Silicon Valley\\
{\tt Omer.Reingold@microsoft.com}}
\date{}
\begin{document}
\maketitle

\begin{abstract}
Given a \dnf\ formula $f$ on $n$ variables, the two natural size measures
are the number of terms or size $s(f)$, and the maximum width of a term
$w(f)$. It is folklore that short \dnf \ formulas can be made
narrow. We prove a converse, showing that narrow formulas can be 
sparsified. More precisely, any width $w$
\dnf\ irrespective of its size can be $\epsilon$-approximated by a
width $w$ \dnf\ with at most $(w\log(1/\epsilon))^{O(w)}$ terms.

We combine our sparsification result with the work
of Luby and Velikovic \cite{LubyV91,LubyV96} to give a faster
deterministic algorithm for approximately counting the number of
satisfying solutions to a \dnf. Given a formula on $n$ variables with
$\poly(n)$ terms, we give a deterministic $n^{\tilde{O}(\log
  \log(n))}$ time algorithm that computes an additive $\epsilon$
approximation to the fraction of satisfying assignments of $f$ for
$\epsilon = 1/\poly(\log n)$. The previous best result due to Luby
and Velickovic from nearly two decades ago had a run-time of
$n^{\exp(O(\sqrt{\log \log n}))}$ \cite{LubyV91,LubyV96}.  
\end{abstract}
\newpage

\section{Introduction}

A natural way to represent a Boolean function $f: \zo^n \rightarrow \zo$
is to write it as a \cnf\ or \dnf\ formula. The class of functions that
admit compact representations of this form (aka polynomial size
\cnf\ and \dnf\ formulae) are central to Boolean function analysis,
computational complexity and machine learning.  

Given a \dnf\ formula $f$ on $n$ variables, the two natural size measures
are the number of terms or size $s(f)$, and the maximum width of a term
$w(f)$. The analogous measures for a \cnf, are the number of clauses and
clause width. It is folklore that every \dnf\ formula $f$ with $m$
terms can be $\eps$-approximated by another \dnf\ $g$ where $s(g) \leq
m$ and $w(g) \leq \log(m/\eps)$, regardless of $w(f)$. 
The formula $g$ is a sparsification of $f$ obtained by simply
discarding all terms of width larger than $\log(m/\eps)$. In other
words, short \dnf\ formulas can be made narrow. An
analogous statement can be derived for \cnf s.

In this work, we show the reverse connection: narrow formulae can be
made short. Indeed, we prove the existence of a strong form of
approximation known as sandwiching approximations which are
important in pseudorandomness. In this work we only consider
approximators which are also Boolean functions.

\begin{Def}
Let $f:\zo^n \rgta \zo$. We say that  functions $f_u,f_\ell:\zo^n \rgta \zo$ are $\eps$-sandwiching
approximators for $f$ if $f_\ell(x) \leq f(x) \leq f_u(x)$ for every
$x \in \zo^n$, and
\begin{align*}
\Pr_{x \in \zo^n}[f_\ell(x) \neq f(x)]  = \Pr_{x \in \zo^n}[(f_\ell(x)=
  0) \wedge (f(x)=1)] \leq \eps, \\
\Pr_{x \in \zo^n}[f_u(x) \neq f(x)]   = \Pr_{x \in \zo^n}[(f_u(x)=
  1) \wedge (f(x)=0)] \leq \eps.
\end{align*}
\end{Def}

Our main result is the existance of $\eps$-sandwiching approximators
for arbitrary width $w$ \dnf s using short width $w$ \dnf s
where the number of clauses depends only on $w$ and $\eps$.

\begin{Thm}
\label{th:mainstruct}
For every width-$w$ \dnf\ formula $f$ and every $\eps > 0$, there
exist \dnf\ formulae $f_\ell, f_u$ each of width $w$ and
size at most $(w\log(1/\epsilon))^{O(w)}$ which are $\eps$-sandwiching
approxmiators for $f$.
\end{Thm}

Our result is proved by  a sparsification procedure for \dnf\ formulae
which uses the notion of quasi-sunflowers due to Rossman
\cite{Rossman10}. The best previously known result along these lines
was due to Trevisan \cite{Trevisan04}, who built on previous work by
Ajtai and Wigderson \cite{AjtaiW85}. Trevisan shows that every width
$w$ \dnf\ has $\eps$-sandwiching approximators that are decision trees
of depth $d = O(w2^w\log(1/\eps))$.

A $k$-junta is a function which depends only on $k$
variables. We say that $f, g: \zo^n \rgta \zo$, we say that $g$
$\eps$-approximates $f$ if 
$$\Pr_{x \in \zo^n}[f(x) \neq g(x)] \leq \eps.$$

A corollary of our result is the following {\em junta theorem} for
\dnf s. 

\begin{corollary}
\label{cor:junta}
Every width-$w$ \dnf\ formula is $\eps$-approximated by
a $(w\log(1/\epsilon))^{O(w)}$-junta.
\end{corollary} 

A similar but incomparable statement can be derived from Friedgut's junta theorem
\cite{Friedgut98}. It is easy to see that width $w$ \dnf s have average
sensitivity at most $2w$ \footnote{\cite{Amano} shows a sharp bound of
  $w$}, so by Friedgut's theorem any width $w$ \dnf\ is $\eps$-close to a
$2^{\tilde{O}(w/\epsilon)}$-junta. Friedgut's result gives better
dependence on $w$, whereas we achieve much better dependence on
$\eps$. Friedgut's approximator is not {\em a priori} a
small-width \dnf, and  one does not get sandwiching
approximations. Trevisan's result implies that any width $w$ \dnf \ is
$\eps$-approximated by a $k$-junta for $k = \exp(O(w2^w\log(1/\eps)))$
\cite{Trevisan04}.

\tref{th:mainstruct} has interesting consequences for other parameter
settings. One example is the following:

\begin{corollary}
\label{cor:logn}
Every width-$O(\log n)$ \dnf\ formula on $n$ variables is
$n^{-O(1)}$ close to a \dnf \ of width $O(\log n)$ and size $n^{O(\log\log(n))}$.
\end{corollary}

In Section \ref{sec:open}, we conjecture that a better bound should be
possible in \tref{th:mainstruct}, which is singly exponential in
$w$. If true, this conjecture will give better bounds for both
Corollaries \ref{cor:junta} and \ref{cor:logn}.

\eat{
For $\epsilon > 0$, any width\footnote{Width is defined as the maximum
  number of literals in any term.} at most $w$ DNF $f:\zo^n \rgta \zo$
is $\epsilon$-strongly approximable by the class of width at most $w$
DNFs with at most $W = (w^350\log(1/\epsilon))^w$ terms.  A
similar statement holds for \cnf s.}

\subsection{\dnf\ Counting and Pseduorandom Generators}

The problem of estimating the number of satisfying solutions
to \cnf\ and \dnf\ formulae is closely tied to the problem of
designing pseudorandom generators for such formulae with short
seed-length. These problems have been studied extensively
\cite{KarpL83, AjtaiW85, NisanW88, Nisan91, LubyV91, LubyV96,
  LubyVW93, Trevisan04, Bazzi09, Razborov09, DeETT10}. 

For a formula $f$, let 
$$\bias(f) = \pr_{x \in \zo^n}[f(x) = 1].$$ 
Given a formula $f$ from a class $\mc{F}$ of functions, the goal of a
counting algorithm for the class $\mc{F}$ is to compute $\bias(f)$. We refer to the counting problems for \cnf s and \dnf s as $\ccnf$ and $\cdnf$ respectively. 
The problem of computing $\bias(f)$ exactly is $\#\Ptime$-hard
\cite{Valiant79}, hence we look to approximate $\bias(f)$. 

An algorithm gives an $\eps$-additive approximation for
$\bias(f)$ if its output is in the range $[\bias(f) -\eps, \bias(f) + \eps]$.
It is easy to see that additive approximations for \cnf s and \dnf s
are equivalent.  There is a trivial solution based on random sampling,
but finding a deterministic polynomial time algorithm has proved challenging. 

Computing multiplicative approximations to $\bias(f)$ is harder, and here the
complexities of $\ccnf$ and $\cdnf$ are very different. An algorithm
is said to be a $c$-approximation algorithm if its output lies in the
range $[\bias(f),c\bias(f)]$. It is easy to see that obtaining a
multliplicative approximation for $\ccnf$ is $\NP$-hard. Karp and Luby
gave the first multiplicative approximation for $\cdnf$, their
algorithm is randomized \cite{KarpL83}. There is a reduction between additive and
multiplicative approximations for $\cdnf$: for $\dnf$ formulae with
$m$ terms, the problem of computing a $(1 + \eps)$-multiplicative
approximation can be reduced deterministically to the problem of
computing an $(\eps/m)$-additive approximation to $\cdnf$. This
reduction is stated explicitly in \cite{LubyV96}, where is attributed
to \cite{KarpL83,KLM:89}

Derandomizing thes Karp-Luby algorithm is an important problem in
derandomization that has received a lot of attention starting form the
work of Ajtai and Wigderson
\cite{AjtaiW85,LN:90,LubyV91,LubyVW93,LubyV96, Trevisan04}. The best
previous result is due to Luby and Velickovic \cite{LubyV91,LubyV96} from
nearly two decades ago: they gave a deterministic
$n^{\exp(O(\sqrt{\log \log n}))}$ time algorithm that can compute an
$\eps$-additive approximation for any fixed constant $\epsilon$.

A natural approach to this problem is to design pseudorandom
generators (\prg s) with small seeds that can $\eps$ fool depth two
circuits. This problem and its generalization to constant depth circuits 
are central problems in pseudorandomness \cite{AjtaiW85,NisanW88,Nisan91,LubyV96, LubyVW93,Trevisan04,Bazzi09,Razborov09,Braverman10,DeETT10}. 

\begin{definition}
A generator $G:\zo^r \rgta \zo^n$ $\delta$-fools a class $\mc{F}$ of
functions if 
\begin{align*}
\left|\E_{y \in \zo^r}[f(G(y))] - \bias(f)\right| \leq \delta
\end{align*}
for all $f\in \mc{F}$. The genrator is said to be explicit if $G$ is
computable in time polynomial in $r$ and $n$.
\end{definition}

A generator with seed-length $r$ that $\eps$-fools \dnf s with $m$ clauses gives
an $\eps$-additive approximation for $\bias(f)$ in $\poly(m,n,2^r)$
time by enumerating over all seeds. Such an algorithm only requires
black-box access to $f$. The reduction form \cite{KarpL83,KLM:89} implies that an
optimal pseduorandom generator for $\dnf$s with seedlength
$O(\log(mn/\eps))$ will give a deterministic multiplicative
approximation algorithm for $\cdnf$. However, the best known generator
currently due to De, Etesami, Trevisan and Tulsiani \cite{DeETT10}
requires seed length $O((\log(mn/\eps)^2)$. The Luby-Velikovic
algorithm is a not a  black-box algorithm, but \prg s for small-width
\dnf s are an important ingredient.

\subsubsection*{Our Results}

We use our sparsification lemma to give a better \prg \ for the class of
width $w$ \dnf\ formulae on $n$ variables, which we denote by  $\dnf(w,n)$.
\footnote{The $\tilde{O}()$ notation is used to hide terms that are
  logarithmic in the arguments.}

\begin{theorem}\label{th:foolswidth}
For all $\delta$, there exists an explicit generator $G:\zo^r \rgta
\zo^n$ that $\delta$-fools $\dnf(w,n)$ and has seed-length  
\begin{align*}
r & = \tilde{O}\left(w^2 + w\log\left(\frac{1}{\delta}\right) +
\log\log(n)\right). 
\end{align*}
\end{theorem}

In comparison, Luby and Velickovic \cite{LubyV96} give a \prg\ with
seed-length $O(2^w  + \log \log n)$ for fooling width $w$ \dnf s. Note
that for $w = O(\log \log n)$ and $\delta$ constant, the seed-length of
the our generator is $\tilde{O}((\log\log n)^2)$, whereas Luby and
Velickovic need seed-length $O(\log^{O(1)}n)$. For $w = \log\log(n)$
and $\delta \geq 1/\poly(n)$, our seed-length is still $\tilde{O}(\log n)$.

The improved generator for small-width \dnf s is obtained  by
using our sparsification result to reduce fooling
width $w$ \dnf s with an arbitrary number of terms to fooling
width $w$ \dnf s with $2^{\tilde{O}(w)}$ terms.  We then apply recent
results by De \etal\ on fooling \dnf\ formulas using small-bias
spaces. The fact that our sparsification gives sandwiching
approximators is critical for this result.

The Luby-Velickovic counting algorithm can be viewed as a (non
black-box) reduction from fooling \dnf s of size $\poly(n)$ to fooling
\dnf s of smaller width. Given \tref{th:foolswidth}, we can
improve and simplify their analysis to get a faster deterministic
counting algorithm. This is the first progress on this well-studied
problem in nearly two decades. In addition, we can allow for smaller
values of $\eps$. 

\begin{theorem}\label{th:mainintro}
There is a deterministic algorithm which when given a \dnf \ formula
on $n$ variables of size $m$ as input, returns an
$O(\eps)$-additive approximation to $\bias(f)$ in time 
\begin{align*}
\left(\frac{mn}{\eps}\right)^{\tilde{O}(\log\log(n) + \log\log(m) + \log(1/\eps))}
\end{align*}
For $m \leq \poly(n)$ and $\eps \geq 1/\poly(\log n)$, the running
time is $O(n^{\tilde{O}(\log\log(n))})$.
\end{theorem}

\newcommand{\zro}{\{*,0,1\}}

H{\aa}stad's celebrated Switching Lemma \cite{Hastad86} is a powerful tool
in proving lower bounds for small-depth circuits. It also has
applications in computational learning \cite{LMN:93,Mansour:95} and
\prg \ constructions \cite{AjtaiW85,GMRTV:12}.  
As an additional application of our sparsification result, we give a
partial derandomization of the switching lemma. The parameters we
obtain are close to that of the previous best results due to Ajtai and
Wigderson \cite{AjtaiW85} and perhaps more importantly, our argument
is conceptually simpler, involving iterative applications of our
sparsification result and a naive union bound. We defer the details to
\sref{sec:slemma}.

\section{DNF Sparsification}
\label{sec:sparsify}

We will consider \dnf\ formulas that are specified as  $f =
\vee_{i=1}^m T_i$ where the representation is minimal in the
following sense:
\begin{itemize}
\item Each $T_i$ is non-constant. Hence each term is non-empty (else
  we replace it by $1$), and does not contain a variable and its negation
  (else we replace it by $0$). This guarantess that $\pr_x[T_i =1] \leq
  1/2$.
\item Each that $T_i$ is not implied by some other $T_j$; if this is so,
we can simply drop $T_i$ from the definition of $f$. This means that
when viewed as a set of literals, $T_j\not\subset T_i$. A consequence
is that $T_i \cap T_j \subsetneq T_j$.
\end{itemize}

If some stage of our sparsification produces a representation which
is not minimal, we can convert it to a minimal represntation without
increasing the number of terms.

We call a DNF $f$ {\sl unate} if it does not contain a variable and its
negation.

\subsection{Sparsification using Sunflowers}

We will first show the following weaker version of
\tref{th:mainstruct} with a bound of $(w2^w\ln(m/\epsilon))^{w}$, and
assumes that $f$ is unate. The proof will illustrate the key ideas
behind our sprsification procedure. 

\begin{Thm}
\label{th:main-simple}
For every unate \dnf\ formula $f$ with width $w$ and size $m$ every
$\eps > 0$, there exist \dnf\ formulae $f_\ell, f_u$ each with width $w$ and
at most $(w\log(m/\epsilon))^{O(w)}$ which are $\eps$-sandwiching
approxmiators for $f$.
\end{Thm}

The starting point of our sparsification result is the Erd\H{o}s-Rado
Sunflower Lemma \cite{ER:60}.

\begin{Def}
Let $k \geq 3$. A collection of subsets $S_1, \ldots,S_k \subseteq
[n]$ is a sunflower with core $Y$ if $Y \subsetneq S_i$ for all $i$
and $S_i \cap S_j = Y$ for all $i \neq j$. The sets $S_i \setminus Y$
are called the petals.
\end{Def}

The set systems that we consider will arise from the terms in some  minimal
representation of a monotone \dnf. This will ensure that the petals
are always non-empty, although the core might be empty.

The celebrated Erd\H{o}s-Rado Sunflower Lemma guarantees that every
sufficiently large set system of bounded size sets contains large
sunflowers. 

\begin{Thm}(Sunflower Lemma, \cite{ER:60})
Let $\mc{F} = \{S_1,\ldots,S_m\}$ be a collection of subsets of $[n]$,
  each of cardinality at most $w$. If $m > w! (k-1)^w$,
  then $\mc{F}$ has a {\sl sunflower} of size $k$.  
\end{Thm}

The lemma and its variants have found several applications
in complexity theory, we refer the reader to \cite[Chapter 7]{Jukna}
for more details. We will use it to prove \tref{th:main-simple}.

\begin{proof}(Proof of \tref{th:main-simple}.)
Fix a unate, width $w$ \dnf\ $f = T_1 \vee T_2 \vee \cdots \vee T_m$ and for
simplicity suppose that $f$ is monotone.
Since $f$ is monotone, we can think of
each term $T_i$ as a set of variables of size at most $w$. Set $k =
2^w \ln(m/\epsilon)$.  Provided 
\begin{align}
\label{eq:cond-1}
m \geq \left(w2^{w}\ln\left(\frac{m}{\eps}\right)\right)^w \geq
w!(k-1)^w
\end{align} 
the Sunflower Lemma guarantees the  existance of a collection of terms
$T_{i_1},\ldots, T_{i_k}$ with a core $Y = \cap_{j=1}^k T_{i_j}$ and
disjoint petals $T_{i_j} \setminus Y$. 
Hence we can write
\begin{align*}
\vee_{j=1}^kT_{i_j} = Y \wedge \left(\vee_{j=1}^k (T_{i_j}\setminus Y) \right) = Y
\wedge g \ \text{where} \ g = \vee_{i=1}^k(T_{i_j} \setminus Y).
\end{align*}
Note that $g$ is a read-once \dnf\ of width $w$ and size $k
=2^w\ln(m/\eps)$, so it is almost surely satisfied by a random assignment: 
\begin{align*}
\pr_x[g(x) = 0] = \prod_{i=1}^k\pr_x[T_{i_j}\setminus Y  = 0] \leq
\left(1 - \frac{1}{2^w}\right)^k \leq \frac{\eps}{m}.
\end{align*}
The first inequality holds because each $T_{i_j}\setminus Y$ is a term
with width at most $w$, and the second by our choice of $k$.

Thus a natural way to get an upper sandwiching approximation is to
replace $g(x)$ by the constant $1$, which is equivalent to replacing  
$\vee_{j=1}^kT_{i_j}$ with $Y$. Let $f': \zo^n \rgta \zo$ be the
\dnf\ formula obtained by this replacement. It is clear that 
$f(x) \leq f'(x)$. Further, 
$$\pr_x[f(x) = 0 \and f'(x)= 1] \leq \pr_x[g(x) =0] \leq \frac{\eps}{m}.$$ 
Finally, we have $s(f') \leq s(f) - (k-1)$.

We can now iteratively apply the above argument as long as the number
of terms is larger than the bound in Equation \eqref{eq:cond-1}. In
each iteration we reduce $s(f)$ by $k-1$. Thus, we
repeat the process at most $m/(k-1)$ times, obtaining an  upper
approximating formula $f_u$ where 
\begin{align*}
f(x) & \leq f_u(x) \ \forall  x \in \zo^n,\\
\pr_x[f(x) \neq f_u(x)] & \leq \frac{m}{k-1} \cdot \frac{\epsilon}{m}  =
\epsilon,\\ 
s(f_u) & \leq \left(w2^{w}\ln\left(\frac{m}{\eps}\right)\right)^w.
\end{align*}

We next describe the construction of the lower approximating
formula $f_\ell$. We start with the  sunflower $T_{i_1}, \cdots,
T_{i_k}$ with core $Y$. Now consider the formula $f''$
obtained from $f$ by dropping one of the terms, say $T_{i_1}$. Then,
$f''(x) \leq f(x)$. Further, the two of them differ only if $f''(x)
=0$ and $f(x) =1$, which happens if $T_{i_1} =1$ whereas $T_{i_j} =0$
for $j \in \{2,\ldots,k\}$. Hence we can bound this probability by
\begin{align*}
\pr_x[f''(x) \neq f(x)] & = \pr_x[T_{i_j} =1]\cdot \pr_x[(\vee_{j=2}^k T_{i_j}) = 0|T_{i_j} =1]\\
& = \frac{1}{2}\pr_x[(\vee_{j=2}^k T_{i_j}\setminus Y) = 0] = \frac{1}{2}\left(1
- \frac{1}{2^w}\right)^{k-1} \leq \frac{\eps}{m}
\end{align*}
where the second inequality holds since by the sunflower property,
conditioning on $T_{i_1} =1$ fixes the core $Y =1$, but does not
affect the other petals. Note that $s(f'') \leq s(f) -1$. We now
iterate this step no more than $m$ times to obtain a formula $f_\ell$ where
\begin{align*}
f_\ell (x) & \leq f(x) \ \forall  x \in \zo^n,\\
\pr_x[f_\ell(x) \neq f(x)] & \leq m \cdot \frac{\epsilon}{m}  =
\epsilon,\\ 
s(f_u) & \leq \left(w2^{w}\ln\left(\frac{m}{\eps}\right)\right)^w.
\end{align*}
\end{proof}

\tref{th:main-simple} is weaker than \tref{th:mainstruct} in the
assumption of unateness, the dependence on $m$ and the dependence on
$w$. We briefly sketch how one can handle the first two issues.

\begin{enumerate}
\item {\bf Unateness.} One can remove
  this assumption by using \lref{lem:monotone-2} which
  guarantees that any \dnf \ formula contains a large sub-formula
  which is unate. The resulting statement already suffices for Corollary
  \ref{cor:logn}, since any width $\log(n)$ \dnf \ can have at most
  $n^{O(\log(n))}$ many clauses.
\item {\bf Dependence on $m$.} The size of the approximators depends
  logarithmically on $m$. One can avoid this by observing that when
  the formula size is large, the error resulting from each step of the
  sparsification is
  tiny. One can use this argument to get a size bound of
  $(2^w\ln(1/\eps))^{O(w)}$ which is independent of $m$. 
\item {\bf Dependence on $w$.} The final bound is exponential in $w^2$ rather than $w$. This
  comes from the $(k-1)^w$ term in the Sunflower Lemma, which we apply
  for $k =2^w$. The question of whether the $w!$ term in the Sunflower Lemma is
  necessary is a well-known open problem in combinatorics. But there
  is a lower bound of $(k-1)^w$ \cite{Jukna}. So even if the lower
  bound were to be right answer, it does not (directly) imply a better
  bound for \tref{th:main-simple}.  
\end{enumerate}

\subsection{Sparsification using Quasi-Sunflowers.}

The main property of the sunflower system we used in \tref{th:main-simple} is that the
formula $g$ on the petals is highly biased towards $1$. As shown by Rossman \cite{Rossman10}, one can guarantee the existence of such
``quasi-sunflower'' systems satisfying this weaker property, even
when the number of terms is much smaller than in the usual sunflower
lemma. We adapt our argument to use quasi-sunflowers
instead of sunflowers, to obtain \tref{th:mainstruct}.

We shall use the notion of quasi-sunflower due to Rossman \cite{Rossman10}.

\begin{Def}(Quasi-Sunflowers, \cite{Rossman10}) A unate \dnf\ formula
  $h = \vee_{i=1}^k T_i$ where $k \geq 2$ is a
  $\gamma$-quasi-sunflower with core $Y = \cap_{j=1}^kT_i$, and petals
  $\{T_i\setminus Y\}_{i=1}^k$ if 
\begin{align*}
\pr_x[\vee_{i =}^k(T_i\setminus Y) =1] \geq 1 - e^{-\gamma}.
\end{align*}
\end{Def}

Quasi-sunflowers extend the notion of a sunflower in the sense that
even though the ``petals'' $(T_{i_j} \setminus Y)$ are not necessarily
disjoint, the probability that none of them is satisfied is
small.  We disallow $k =1$, since otherwise every term is trivially a
quasi-sunflower. Since we insist that no term of a \dnf\ is contained
in another, the petals are non-empty. Hence each petal is satified
with probability at most $1/2$, so every $\gamma$-sunflower has $k =
\Omega(\gamma)$ petals. 

\begin{lemma}
\label{lm:sunflowercnf}
(Quasi-Sunflower Lemma, \cite{Rossman10}) Any unate width $w$
\dnf\ formula with $m$ terms contains a $\gamma(m)$-quasi-sunflower where
\begin{align}
\label{eq:gamma}
\gamma(m) := \frac{1}{5}\left(\frac{m}{w!}\right)^{1/w}.
\end{align}
\end{lemma}

Rossman states the result in the language of set systems, which we
have rephrased in the language of \dnf s. We show the equivalence of the
two in the appendix. 

The following lemma will be used to analyze a single step of our sparsification. 

\begin{Lem}\label{lem:monotone-1}
Let $g = \vee_{i=1}^m T_i$ be a unate \dnf. Then
\begin{align*}
\pr_x[(T_1 =1) \wedge ((\vee_{i=2}^k T_i) =0)] \leq
\pr_x[(\vee_{i=1}^k T_i) =0].
\end{align*}
\end{Lem}
\begin{proof}
Without loss of generality suppose that $g$ is monotone. Since every
term in $g$ is also monotone, Kleitman's lemma \cite[Chapter 6]{AlonS}
implies that
\begin{align*}
\pr_x[(T_1 =0) \wedge ((\vee_{i=2}^k T_i) =0)] \geq \pr_x[T_1 =0]\cdot
\pr_x[(\vee_{i=2}^k T_i) =0]\\
\pr_x[(T_1 =1) \wedge ((\vee_{i=2}^k T_i) =0)] \leq \pr_x[T_1 =1]\cdot
\pr_x[(\vee_{i=2}^k T_i) =0]
\end{align*}
Hence we have
\begin{align*}
\frac{\pr_x[(T_1 =0) \wedge ((\vee_{i=2}^k T_i) =0)]}{\pr_x[T_1 =0]}
\geq 
\pr_x[(\vee_{i=2}^k T_i) =0] \geq 
\frac{\pr_x[(T_1 = 1) \wedge ((\vee_{i=2}^k T_i) =0)]}{\pr_x[T_1 =1]}.
\end{align*}
But this implies that
\begin{align*}
\pr_x[(T_1 = 1) \wedge ((\vee_{i=2}^k T_i) =0)] \leq
\pr_x[(\vee_{i=1}^k T_i) =0]\cdot \frac{\pr_x[T_1 =1]}{\pr_f[T_1 =0]}
\leq \pr_x[(\vee_{i=1}^k T_i) =0]
\end{align*}
where the last inequality  follows because for any (non-empty) term
$T$, 
\begin{align}
\label{eq:T1}
\pr_x[T =1] \leq \frac{1}{2} \leq \pr_x[T = 0].
\end{align}
\end{proof}

The only property of $T_1$ that we use is that $\pr_x[T_1 =1]
\leq \pr_x[T_1 = 0]$. Indeed, we can drop any set of terms $\{T_i\}_{i
  \in S}$ which satisfies  $\pr_x[\vee_{i \in S}T_i =1] \leq
\pr_x[\vee_{i \in S}T_i= 0]$. 


The following is our key technical lemma. It applies to unate
formulae and allows us to reduce the size of formula by (at least) $1$.
 
\begin{Lem}
\label{lem:unate-main}
For every unate width-$w$ \dnf\ formula $g$ of size $m$, there
exist width-$w$ \dnf\ formulae $g_\ell, g_u$ each of size 
at most $m -1$ that are $e^{-\gamma(m)}$ sandwiching
approximators for $g$. 
\end{Lem}
\begin{proof}
Let $g = \vee_{i=1}^m T_i$. \lref{lm:sunflowercnf} guarantees the
existance of a $\gamma(m)$-quasi-sunflower  $h = \vee_{i=1}^kT_{i_j}$
where $\gamma(m)$ is given by Equation \eqref{eq:gamma}. Letting $p(x) =
\vee_{i=1}^k(T_{i_j} \setminus Y)$ be the formula on the petals, we
have $\pr_x[p(x) =0] \leq e^{-\gamma(m)}$. We can write
\begin{align*}
h(x) \ = \ \vee_{j=1}^kT_{i_j} \ = \ Y \wedge \left(\vee_{j=1}^k
(T_{i_j}\setminus Y) \right) \ = \ Y \wedge p(x)
\end{align*}

We get an upper sandwiching \dnf\ formula $g_u: \zo^n \rgta \zo$ from
$g(x)$ by replacing $p(x)$ by the constant $1$, which is equivalent to
replacing $h(x)$ with the core $Y$. It is clear that 
\begin{align*}
g(x) \leq g_u(x), \ \ s(g_u) \leq s(g) - (k-1) \leq s(g) -1.
\end{align*}
Further, 
\begin{align*}
\pr_x[g(x) \neq g_u(x)] & = \pr_x[(g(x) = 0) \wedge (g_u(x)= 1)]\\ 
& \leq \pr_x[p(x) =0] \\
& \leq e^{-\gamma(m)}.
\end{align*}  

We now construct the lower sandwiching approximation.  Let $g_\ell$ be
the formula obtained from $g$ by dropping the term $T_{i_1}$. Then, it
is clear that
$$g_\ell(x) \leq g(x), \ \ s(g_\ell) \leq s(g) -1.$$
Further, 
\begin{align*}
\pr_x[g(x) \neq g_\ell(x)] & = \pr_x[g(x) =1 \wedge g_\ell(x) = 0]\\ 
& \leq \pr_x[ ((T_{i_1}\setminus Y) =1)
  \wedge (\vee_{j=2}^k(T_{i_j}\setminus Y) ) = 0]\\
& \leq \pr_x[p(x) = 0]  \hspace{2cm} \text{(By \lref{lem:monotone-1})}\\
& \leq e^{-\gamma(m)}. 
\end{align*}
\end{proof}

One can prove Theorem \ref{th:mainstruct} for unate \dnf s by
repeated applications of this Lemma. To handle the general case, we
use the following simple lemmas to reduce the problem of
constructing sandwiching approximations to the unate case. 

\begin{Lem}
\label{lem:or}
Let $f,g,h:\zo^n \rgta \zo$ be such that  $f = g
\vee h$. Let $g_\ell, g_u$ be $\eps$-sandwiching approximators for
$g$. Then $g_\ell \vee h$ and $g_u \vee h$ are $\eps$-sandwiching
approximators for $f$. 
\end{Lem}
\begin{proof}
It is easy to see that for every $x \in \zo^n$,
$$g_\ell(x)\vee h(x) \leq g(x)\vee h(x) \leq g_u(x)\vee h(x).$$
We bound the approximation error for $g_\ell
\vee h$, the proof for $g_u \vee h$ is similar.
\begin{align*}
\pr_x[(g_\ell(x)\vee h(x)) \neq (g(x)\vee h(x))] & = \pr_x[(g_\ell(x)\vee
  h(x) = 0) \wedge(g(x)\vee h(x) =1)] \\
& = \pr_x[(g_\ell(x) =0) \wedge (g(x) = 1) \wedge(h(x) =0)] \\
& \leq \pr_x[(g_\ell(x) =0) \wedge (g(x) = 1)] \\
& \leq \eps.
\end{align*}
\end{proof}

\begin{Lem}\label{lem:monotone-2}
For every width $w$ \dnf\ $f = \vee_{i=1}^mT_i$ of size $m$, there
exists $S \subseteq [m]$ where $|S| \geq m/2^w$ such that
the formula $g = \vee_{j \in S} T_{i_j}$ is unate.
\end{Lem}
\begin{proof}
Pick a random set of literals $S$ as follows: for
each of the variables $x_i$ add one of $x_i$ or $\bar{x}_i$ to $S$
uniformly at random. Let $g_S$ be the sub-formula of $f$ formed of
terms containing only literals from $S$. Then, $g_S$ is always unate. 

Each term has at least a $2^{-w}$ chance of being in $g_S$. By
linearity of expectation  
$$\ex_S[s(g_S)] \geq \frac{m}{2^w}.$$ 
\end{proof}

We will use the following asymptotic bound whose proof is a
calculation and is deferred to the appendix. 
\begin{fact}\label{fct:tedious}
For $\gamma:\R_+ \rgta \R_+$ defined by \eref{eq:gamma}, $W =
(2w)^{3w}(50 \log(1/\epsilon))^w$, and $\epsilon \leq 1/4$, 
$$\sum_{j=W +1}^m e^{-\gamma(j/2^w)} \leq \eps.$$
\end{fact}

We can now prove \tref{th:mainstruct}:
\begin{proof}
Let $f= \vee_{i=1}^m T_i$. By applying \lref{lem:monotone-2}, we can
write $f = g \vee h$ where $g$ is unate and has $m' \geq m/2^w$
terms. By \lref{lem:unate-main}, there exist sandwiching approximators $g_\ell,g_u$ each of width $w$ and size at most $m'-1$, whose error is bounded by
$$ e^{-\gamma(m')} \leq e^{-\gamma(m/2^w)}.$$ 
By \lref{lem:or}, $f^1_\ell = g_\ell \vee h$ and $f^1_u= g_u \vee h$ are
$e^{-\gamma(m')}$ sandwiching approximations for $f$. Further
$$s(f^1_\ell) = s(g_\ell) +s(h) \leq s(g) -1 +s(h) \leq s(f) -1$$
and similarly $s(f^1_u) \leq s(f) -1$. 

We iterate this construction separately for the upper and lower
approximator till the size of the formulae drops below $W$. This gives
the sequence
\begin{align*}
f(x) \leq f_u^1(x) \cdots \leq f_u^{k_u}(x) := f^u(x)\\
f(x) \geq f_\ell^1(x) \cdots \geq f_\ell^{k_\ell}(x) := f_\ell(x)
\end{align*}
where $s(f_\ell), s(f_u) \leq W$.
We can bound the error of these approximators by
\begin{equation}
\sum_{j=W +1}^m e^{-\gamma(j/2^w)} \leq \eps.
\end{equation}
where the inequality is from Fact \ref{fct:tedious}.
This completes the proof of \tref{th:mainstruct}.
\end{proof}

\section{Fooling Small-Width DNFs}
\label{sec:prg}

We next use our sparsification result to construct a pseudorandom
generator for small-width \dnf s, obtaining an exponential improvement
in terms of the width over the generator of Luby and Velickovic
\cite{LubyV96}. We restate \tref{th:foolswidth} with
the exact asymptotics for $r$.

\begin{theorem}\label{th:foolswidth2}
For all $\delta$, there exists an explicit generator $G:\zo^r \rgta
\zo^n$ that $\delta$-fools all width $w$ \dnf s and has seed-length  
\begin{align*}
r & = O\left(w^2\log^2(w) + w\log(w)\log\left(\frac{1}{\delta}\right)
+ \log\log(n)\right)
\end{align*}
\end{theorem}


We prove the theorem as follows: we first use our sparsification
result to reduce the case of fooling width $w$ \dnf s with an
arbitrary number of terms to that of fooling width $w$ DNFs with
$2^{\tilde{O}(w)}$ terms and then apply the recent results due to De \etal~\cite{DeETT10} showing that small-bias spaces fool DNFs with few terms.

\begin{definition}[$k$-wise $\epsilon$-biased spaces]
 A distribution $\calD$ over $\zo^n$ is said to be
 $(k,\epsilon$)-biased space if for every non-empty subset $I
 \subseteq [n]$ of size at most $k$,
\begin{align*}
\left|\pr_{x \lfta \calD}[ \oplus_{i \in I} x_i = 1] -
\frac{1}{2}\right| \leq \epsilon.
\end{align*} 
\end{definition}
Naor and Naor \cite{NaorN93} constructed explicit $(k,\epsilon)$-biased spaces that
require only $O(k + \log(1/\epsilon) + \log\log n)$ bits to sample
from.  

Next, we need the following result of De \etal ~\cite{DeETT10}
showing that $(k,\eps)$-biased spaces fool \dnf s for suitable choices
of $k$ and $\epsilon$.  

\begin{theorem}\cite[Theorem 4.1]{DeETT10}
\label{th:foolsmallbias}
For every $\delta > 0$, every \dnf \ 
with width $w$ and size $m$ is $\delta$-fooled by
$(k,\epsilon)$-biased distributions for 
\begin{align*}
k  & = O\left(w\log\left(\frac{m}{\delta}\right)\right), \\
\log\left(\frac{1}{\epsilon}\right) & =
O\left(w\log(w)\log\left(\frac{m}{\delta}\right)\right). 
\end{align*}
\end{theorem}

De \etal \ prove the above statement only for the case of $k = n$, and they use the bound $w \leq \log(m/\delta)$. Their proof proceeds by
constructing small $\ell_1$-norm sandwiching approximators. The above
statement is obtained by repeating their proof keeping $w$ and $m$
separate, and bounding both the degree and the $\ell_1$ norm of the
resulting approximators. It is easy to see from their proof 
that the approximators have  degree $k \leq O(w\log(m/\delta))$ and
$\ell_1$-norm bounded $(m/\delta)^{O(w\log(w))}$. 
 
We use the fact that to fool a class of functions, it suffices to
fool sandwiching approximators \cite{BGGP,Bazzi09}. 

\begin{fact}\label{fct:foolsandwich}
Let $\mc{F}, \mc{G}$ be classes of functions such that every $f \in
\mc{F}$ has $\eps$-sandwiching approximators in $\mc{G}$. Let $G:
\zo^r \rgta \zo^n$ be a pseudorandom generator that $\eps$-fools
$\mc{G}$. Then  $G$ $(\eps + \delta)$-fools $\mc{F}$.
\end{fact}

We are now ready to prove the main result of this section.

\begin{proof}[Proof of \tref{th:foolswidth2}]

Recall that $\dnf(w,n)$ denotes the class of all width $w$ \dnf\ s on
$n$ variables. Let $\mc{G} \subset \dnf(w,n)$ denote the subset of all formulae
with size at most $m = (w\log(1/\delta))^{cw}$ for some sufficiently large
constant $c$. By Theorem \ref{th:mainstruct}, every $f \in \dnf(w,n)$ can
has $\delta$-sandwiching approximators in $\mc{G}$. 

Next, we apply \tref{th:foolsmallbias} with $m =
(w\log(1/\delta))^{cw}$. Note that
\begin{align*}
\log\left(\frac{m}{\delta}\right) = O\left(w\log(w) +
\log\left(\frac{1}{\delta}\right)\right).  
\end{align*} 
So we conclude that $(k, \eps)$-biased
distributions $\delta$-fool $\mc{G}$ where
\begin{align*}
k & = O\left(w^2\log(w) + w\log\left(\frac{1}{\delta}\right)\right)\\
\log\left(\frac{1}{\eps}\right) & = O\left(w^2\log^2w +
w\log(w)\log\left(\frac{1}{\delta}\right)\right).
\end{align*}
Note that we can sample from such a distribution using a seed of
length
\begin{align*}
r & = O\left(k + \log\left(\frac{1}{\eps}\right) + \log\log(n)\right)\\ 
& = O\left(w^2\log^2(w) + w\log(w)\log\left(\frac{1}{\delta}\right)
+ \log\log(n)\right)
\end{align*}
Finally, by Fact \ref{fct:foolsandwich}, such distributions $2\delta$
fool the class $\dnf(w,n)$.
\end{proof}

\section{Deterministic Counting for DNFs}
\label{sec:dnf-count}

We now use the \prg\ for small-width \dnf s from the previous section
in the Luby-Velickovic counting algorithm \cite{LubyV96}.  The
better seed-length means that we do not need to balance various
parameters as carefully, and can redo their arguments with simpler and
better settings of parameters.   

The input to our algorithm is a $\dnf$ formula $f = \vee_{j=1}^m T_j$
on $n$ variables with size $m$ and width $w$, and the output is an
$\eps$-additive approximation to $\bias(f)$. We set the following parameters
\begin{align*}
k := \log\left(\frac{w}{\epsilon}\right), \  t := \frac{w}{k}, \ w' =
6k, \delta =\frac{\eps}{t}
\end{align*}

Let $\hh = \{h:[n] \rgta [t]\}$ be a family of $k$-wise independent
hash functions. Fix a  hash function $h \in \hh$ and let $B_j = \{i:
h(i) = j\}$. We say the term $T_i$ bad for $h$ if 
$$\max_{j \in [t]}|B_j \cap T_i| > w'$$ 
where we view $T_i$ as a set of variables. Let $f_h$ be the
formula obtained from $f$ by dropping all terms that are bad for $h$.  

Let $G:\zo^r \rgta \zo^n$ be the generator from \tref{th:foolswidth} that fools
$\dnf(w',n)$ with error at most $\delta$. Define a new generator
$G_h:(\zo^{r})^t \rgta \zo^n$ as follows: 
\begin{equation}
  \label{eq:mz}
  G_h(z_1,\ldots,z_t) = x, \text{ where for $j \in [t]$, } x_{|B_j} = G(z_j).
\end{equation}
Thus $G_h$ applies an independent copy of $G$ to each 
bucket defined by the hash function $h$. 

We now state the counting algorithm:

\algo{Algorithm {\sf DNFCount}}{

For each $h \in \hh$,\\
$~~~$ Drop all bad terms for $h$ from $f$ to obtain $f_h$.\\
$~~~$ By enumeration over all $z \in \zo^{rt}$, compute
\begin{align}
\label{eq:ph}
p_h = \pr_{z \in \zo^{rt}}[f_h(G_h(z)) =1].
\end{align}
Return $p_\hh = \max_{h \in \hh}p_h$.
}

We need the following lemma about $k$-wise independent hash functions.
\begin{lemma}\label{lm:hashing}
  Let $\hh:[n] \rgta [t]$ be a $k$-wise independent family of hash
  functions. Then, for every set $S \subseteq [n]$ of size $|S| \leq k
  t$, and every $j \in [t]$,
\[ \pr_{h \in_u \hh}\left[ \,|h^{-1}(j) \cap S| \geq 6k\,\right] \leq 2^{-k}.\] 
\end{lemma}
\begin{proof}
Fix $j \in [t]$. Let $S = \{1,\ldots,kt\}$ without loss of generality.  Let
$\{X_i\}_{i=1}^{kt}$ be indicator random variables that are $1$ if
$h(i) = j$ and $0$ otherwise. Then
\begin{align*}
\ex_{h \in \hh} \left[ \sum_{I \subseteq S, |I| = k} \prod_{i
    \in I} X_i \right] \leq \binom{kt}{k} \cdot \frac{1}{t^k} \leq e^k.
\end{align*}
Applying Markov's inequality,
\begin{align*}
\pr_{h \in_u \hh}\left[ \,|h^{-1}(j) \cap S| \geq 6k\,\right] \leq \frac{e^k}{
\binom{6k}{k}} \leq 2^{-k}.
\end{align*}
\end{proof}

Our analysis requires two Lemmas from \cite{LubyV96}. Since their
terminology and notation differs from ours, we provide proofs of both
these Lemmas in Appendix \ref{app:LV}. 

The first Lemma relates the bias of $f_h$ with that of $f$.

\begin{Lem}\cite[Lemma 11]{LubyV96}
\label{lem:lv1}
We have
\begin{align*}
\forall h \in \hh, \ \ \bias(f_h) \leq \bias(f),\\
\ex_{h \in \hh}[\bias(f_h)] \geq \bias(f_h)-  \epsilon.
\end{align*}
\end{Lem}


The next lemma showing that $G_h$ fools the formula $f_h$ is
essentially \cite[Lemma 7]{LubyV96}. Recall that by Equation \eqref{eq:ph},
$p_h$ is the bias of $f_h$ under distribution generated by $G_h$.

\begin{Lem}\cite[Lemma 7]{LubyV96} 
\label{lem:lv2}
We have $|p_h - \bias(f_h)| \leq \eps$.
\end{Lem}

With these Lemmas in hand, we now analyze the algorithm.

\begin{theorem}
\label{thm:main}
Algorithm {\sf DNFCount}  when given a \dnf\ on $n$
variables with width $w$ and size $m$ as input, returns an
$O(\eps)$-additive approximation to $\bias(f)$ in time
\begin{align*}
O(n^{O(\log(w/\eps))} (\log n)^{O(w)}2^{\tilde{O}(w\log(1/\eps))}m).
\end{align*}
\end{theorem}
\begin{proof}
The correctness of the algorithm is easy to argue. For every $h \in \hh$, 
\begin{align*}
p_h & \leq \bias(f_h) + \eps \ \ \text{(By \lref{lem:lv2})} \\
& \leq \bias(f) + \eps \ \ \text{(By \lref{lem:lv1})}
\end{align*}
Further by \lref{lem:lv1}, there exists $h \in \hh$ such that
\begin{align*}
\bias(f_h) &  \geq \bias(f) - \eps,
\end{align*}
hence by \lref{lem:lv2},
\begin{align*}
p_h & \geq \bias(f_h) -\eps \geq \bias(f) - 2\eps.
\end{align*}
Thus $p_\hh$ is a $2\eps$-additive approximation $\bias(f)$.

We now bound the running time.
Computing $f_h$ for any $h \in \hh$ and evaluating it on $G_h(z)$ for
$z \in \zo^{rt}$ can be done in time $O(mn)$. Thus the running time is
dominated by $|\hh|2^{rt}$. By standard constructions of $k$-wise
independent hash functions, 
\begin{align*}
|\hh| \leq n^{O(k)}.
\end{align*}
Next we bound the seed-length $r$.  Recall that
\begin{align*}
k = \log\left(\frac{w}{\eps}\right),  \delta = \frac{\eps}{t} = \frac{k\eps}{w}\\
\text{Hence} \ \ \log\left(\frac{1}{\delta}\right)  =
\log\left(\frac{w}{\eps k}\right) = k- \log(k).
\end{align*}
Further, $w' = 6k$. Hence \tref{th:foolswidth2},
\begin{align*}
r & = O\left(w'^2\log^2(w ')+
w'\log(w')\log\left(\frac{1}{\delta}\right) + \log\log(n)\right)\\
& = O(k^2\log^2(k) + \log\log(n))\\
rt & = O\left(\frac{w}{k}(k^2\log^2(k) + \log\log(n))\right)\\ 
& = O(wk\log^2k + w\log\log(n)).
\end{align*}
So we get
\begin{align*}
|\hh|2^{rt} \leq \exp(O(k \log(n) + wk\log^2k + w\log\log(n))). 
\end{align*}
Overall the runtime is bounded by
\begin{align*}
O(mn)|\hh|2^{rt} & = \exp(O(\log(w/\eps)\log(n) +
w\log(w/\eps)(\log\log(w/\eps))^2 + w\log\log(n) + \log(m)))\\
& =n^{O(\log(w/\eps))} (\log n)^{O(w)}2^{\tilde{O}(w\log(1/\eps))}m.
\end{align*}
\end{proof}

\tref{th:mainintro} is obtained from \tref{thm:main} by
setting parameters appropriately.

\begin{proof}(Proof of \tref{th:mainintro}.)
Given a \dnf\ formula with size $m$, we can ignore all terms of
width larger than $\log(m/\eps)$ while only changing the bias by
$\eps$. Plugging in $w = \log(m/\eps)$, we can bound the running time by
\begin{align*}
\left(\frac{mn}{\eps}\right)^{\tilde{O}(\log\log(n) + \log\log(m) + \log(1/\eps))}
\end{align*}

For $m = \poly(n), \eps = 1/\poly(\log n)$, this gives
$n^{\tilde{O}(\log\log(n))}$.
\end{proof}


\section{A Derandomized Switching Lemma}
\label{sec:slemma}

H{\aa}stad's celebrated Switching Lemma \cite{Hastad86} is a powerful tool
in proving lower bounds for small-depth circuits. It also has
applications in computational learning \cite{LMN:93,Mansour:95} and
\prg \ constructions \cite{AjtaiW85,GMRTV:12}. This lemma builds on earlier
work due to Ajtai \cite{Ajtai}, Furst, Saxe and Sipser \cite{FSS:84}
and Yao \cite{Yao:85}. 

To state the Switching lemma, we need to set up some notation.
We start with some notation. Given $L \subseteq [n]$ and $x \in \zo^{[n]\setminus L}$ define a 
  restriction $\rho := \rho_{L,x} \in \zro^n$ by $\rho_i = *$ if
$i \in L$ and $\rho_i = x_i$ otherwise. We call the set $L \equiv
L(\rho)$ as the set of ``live'' variables. For $f:\zo^n \rgta \zo$,
and $\rho \in \zro^n$, define $f_\rho:\zo^{L(\rho)} \rgta \zo$ by
$f_\rho(y) = f(x)$, where $x_i = y_i$ for $i \in L(\rho)$ and $x_i =
\rho_i$ otherwise.  

Given a distribution $\calD$ on $2^{[n]}$, let $\calD$ (abusing
notation, the meaning will be clear from context) denote the
distribution on $\rho \in \zro^n$ by setting $\rho = \rho_{L,x}$ where
$L \lfta \calD$ and $x \in_u \zo^{[n]\setminus L}$. Call a
distribution $\calD$ as above $p$-regular if for each $i \in [n]$,
$\pr_{L \lfta \calD}[i \in L] = p$. Let $\calD_p(n)$ (we omit
$n$ if clear from context) denote the $p$-regular distribution on
subsets $L$ of $[n]$ where each element $i \in [n]$ is present in $L$
independently with probability $p$. For $f: \zo^n \rgta \zo$, let
$\dt{f}$ denote the minimum depth of a decision tree computing $f$.

\begin{theorem}[Switching Lemma, \cite{Hastad86}]\label{th:slemmam}
Let $f:\zo^n \rgta \zo$ be a \dnf\ of width $w$ and let $\rho \lfta
\calD_p(n)$. Then, 
$$\pr[\dt{f_\rho} \geq s] < (5pw)^s.$$
\end{theorem}

There has been work on finding a {\sl derandomized} version of the switching
lemma, motivated by better \prg \ constructions .  Such a lemma would
choose the set of live variables in a pseudorandom way, as in
\cite{AjtaiW85}. One could even ask for a stronger derandomization
where the assignments to the non-live variables are also chosen
pseudorandomly, this is done in \cite{GMRTV:12}. We limit ourselves to
the former case here. 


Derandomized switching lemmas were first studied in the seminal work
of Ajtai and Wigderson \cite{AjtaiW85}, with the aim of constructing
better \prg s for constant depth circuits.

\begin{theorem}[\cite{AjtaiW85}]
For all $\gamma \in (0,1]$, $p < 1/n^\gamma$, there is a $p$-regular
  distribution $\calD$ on $2^{[n]}$ with $L \lfta \calD$ samplable
  using $O_\gamma(\log n)$ random bits, and $k = O_\gamma(1)$ such
  that for $\rho \lfta \calD$, and any polynomial size \dnf\ $f$, 
\begin{align*}
\pr[f_\rho \text{ is not a } k\text{-junta}] \leq 1/\poly(n).
\end{align*}
\end{theorem} 

A very recent result along these lines is due to the authors together
with Trevisan and Vadhan, which gives a near-optimal derandomization
in the special case of read-once \dnf s \cite{GMRTV:12}. They use this
to give near \prg s for read-once \dnf s with seed-length
$\tilde{O}(\log n)$.

We remark that if instead of finding a small set of restrictions that
work for all formulas $f$, we are given the formula $f$ as input,
Agrawal et al.~\cite{AgrawalAIPR01} give a polynomial-time algorithm
to find a restriction that simplifies the formula as well as the
bounds given by the switching lemma \tref{th:slemmam}. 

\subsection{Our Result}

We give a different argument that essentially recovers the
result of Ajtai and Wigderson and further gives a trade-off between
the survival probability $p$, the complexity of the restricted
function and the failure probability of the restriction. Our argument
is through repeated applications of \tref{th:mainstruct} and it seems to us
to be simpler than those of Hastad \cite{Hastad86} and Ajtai and Wigderson
\cite{AjtaiW85}.

\begin{theorem}\label{th:slemmad}
There exists a constant $C$ such that for any $w,s, \delta > 0$  and
all $p$ such that
\[p \leq \frac{\delta }{(w\log(1/\epsilon))^{C\log w}},\] 
there is a $p$-regular distribution $\calD$ on $2^{[n]}$ that can be sampled efficiently using  $r$ random bits where
\begin{align*}
r = r(n,s,\epsilon,\delta) = O\left( (\log w) \cdot \left(\log n +
s\log(1/\delta)\right) + w\log(w\log(1/\epsilon))\right),
\end{align*}
the indicator events $\mathsf{1}\{i \in L\}$ are
$p$-biased and the following holds: for any width $w$ \dnf\ $f:\zo^n
\rgta \zo$, and $\rho \lfta \calD$,  
\begin{align*}
\pr[f_\rho \text{ does not have $\epsilon$-sandwiching approximations
    in }\dnf(s,n)]< \epsilon + \delta^{s/4}
\end{align*}
\end{theorem}

In particular, by setting $\delta = 1/n^\gamma$, $s =
\Theta(1/\gamma)$, $\epsilon = 1/\poly(n)$, $w = O(\log n)$, we almost
recover the derandomized switching lemma of Ajtai and Wigderson, with
the main difference being that we need $O((\log n)(\log \log n))$ bits
to sample from $\calD$ and we only get $f_\rho$ has sandwiching
approximations by width $O_\gamma(1)$ \dnf s.   




Our derandomization is based on the intuition that the switching lemma
is easy to show when the number of terms in the original \dnf\ $f$ is
small. For instance, let $f = \vee_{j=1}^{2^w}T_j$
be a width $w$ \dnf. Note that for $0 < p < 1$, and $\rho \lfta \calD_p$, the probability that a single term $T_i$ survives the restriction $f_\rho$ (is not set to be a constant) is at most 
\begin{align*}
\sum_{i=1}^w \binom{w}{i} p^i \left(\frac{(1-p)}{2}\right)^{w-i} \leq
\left( \frac{1+p}{2}\right)^w.
\end{align*}
In particular if $p \leq 1/w$, the above probability is at most
$e/2^w$. Thus, by linearity of expectation, the expected number of
terms that survive the restriction is at most $O(1)$. Hence, by
Markov's inequality, the restricted \dnf\ $f_\rho$ has very few
surviving terms with high probability. Further, as we are only using
Markov's inequality, the above argument would work even if the
restriction $\rho$ is sampled from a distribution where the choices
for different variables are only $k$-wise independent for $k = O(w)$.  

We use \tref{th:mainstruct} to reduce the case of arbitrary \dnf s of
small-width to that of \dnf s with a small number of terms and then
use an argument similar to the above. Unfortunately, the bound in
\tref{th:mainstruct} is not sufficiently strong, so we need to use
somewhat stronger restrictions where the survival probability is $p = w^{-r}$ 
for $r\geq 1$. Such a restriction can be viewed as a sequence of $r$ rounds of
random restrictions, leaving with a $1/w$ fraction of live
variables. We argue that in each round, the width of the formula
decreases by $1/2$ with high probability and then iteratively apply
the argument to the new width $w/2$ formulas.  After $O(\log w)$
rounds, the width reduces to a constant. This corresponds to a random
restriction where the probability of being alive is
$\exp(-\Omega(\log^2 w))$. Moreover, this argument works even when the
random restrictions only have limited independence, yielding
\tref{th:slemmad}.

For $k \leq n$, let $\calD_p(k)$ denote the class of $p$-regular
distributions on $2^{[n]}$ such that for $L \lfta \calD \in
\calD_p(k)$, $\pr[I \subseteq L] \leq 2 p^{|I|}$ for all $I \subseteq [n], |I| \leq k$. There exist explicit distributions
     $\calD \in \calD_p(k)$ that can be sampled using $O(k\log(1/p) +
     \log n)$-random bits. For instance, one can use $p^k$-almost
     $k$-wise independent $p$-biased variables from \cite{NaorN93}.

\begin{claim}\label{clm:slemma}
  There exists a constant $c < 1$ such that the following holds for
  all $\delta,\epsilon > 0$, $ 0 < s \leq w$ and 
\[ p\leq p(w,s) := \frac{c\delta^{s/2w}}{(w^3\log(1/\epsilon))^2}\] 
For any width $w$ \dnf\ $f:\zo^n \rgta \zo$ and $\rho \lfta \calD \in \calD_p(w)$,
  with probability at least $1-\delta^{s/4} - \epsilon$ there exist
  width $w/2$ \dnf s $f_\rho^\ell, f_\rho^u:\zo^{L_\rho} \rgta \zo$ that are
  $\eps$-sandwiching approximators for $f_\rho$.
\end{claim}
\begin{proof}[Proof of \clref{clm:slemma}]
Let $f^\ell, f^u$ be width $w$ \dnf s with at most $h(w) =
w^{3w}(C\log(1/\epsilon))^{w}$ terms that are
$\epsilon^2/2$-sandwiching approximators for $f$ as guaranteed by
\tref{th:mainstruct} for $C$ a large constant.   Consider a random
restriction $\rho$ sampled from a distribution in
$\calD_p(w/2)$. Then, the probability that a fixed term of $f^\ell$
has more than $w/2$ live variables under $\rho$ is at most $2^w \cdot
p^{w/2}$. Therefore, by a union bound, the probability that
$f^\ell_\rho$ has width more than $w/2$ is at most $h(w) 2^w p^{w/2} <
\delta^{s/4}/2$ for a sufficiently small constant $c$. Similarly, the
probability that $f^u_\rho$ has width more than $w/2$ is at most
$\delta^{s/4}/2$.  

Note that as $f^\ell \leq f \leq f^u$, $f^\ell_\rho \leq f_\rho \leq
f^u_\rho$. We now need to show that $f^\ell_\rho, f^u_\rho$ are close
to $f_\rho$ with high probability. Let $\rho \equiv \rho_{L,x}$ and
consider a fixing of the set of live variables $L$. Then as
$f^\ell,f^u$ are $\eps^2/2$-sadwiching approximators for $f$,
\begin{align*}
\ex_{x \in \zo^{[n]\setminus L}}[ \bias(f_{\rho})] & =  \bias(f) \\
& \leq  \bias(f^\ell) + \frac{\epsilon^2}{2} \\ 
& = \ex_{x \in \zo^{[n]\setminus L}}[\bias(f^\ell_{\rho})] + \frac{\epsilon^2}{2}.
\end{align*}
Therefore, 
\[\ex_{x \in \zo^{[n] \setminus L}}[ \bias(f_\rho) -
  \bias(f^\ell_\rho)] \leq \frac{\epsilon^2}{2}.\]
Thus, by Markov's inequality, 
\[ \pr_{x \in \zo^{[n] \setminus L}}[ \bias(f_\rho) -
  \bias(f^\ell_\rho) \geq \epsilon] \leq \frac{\epsilon}{2}.\] 
Using a similar argument to $f^u$, and a union bound, we get that
$f_\rho$ is $\epsilon$-sandwiched  by $(f^\ell_\rho,
f^u_\rho)$ with probability at leat $1 - \delta^{s/4} - \epsilon$. 
\end{proof}

We now prove \tref{th:slemmad}.
\begin{proof}[Proof of \tref{th:slemmad}]

Let $t$ be such that $w/2^t = s$ (we ignore the minor technicality of $t$ being non-integral) and for $r = 1,\ldots,t$, let $p_r = p(w/2^r,s)$ as defined in the above claim. For $i \in [t]$, let $L_i$ be chosen independently from a distribution in $\calD_{p_i}(w/2^i)$. Let $L = \cap_{i=1}^t L_i$ and for $x \in_u \zo^n$, let $\rho = \rho_{L,x}$. Then, $\rho$ is a $q$-regular random restriction with 
\begin{align*}
q = \prod_{i=1}^tp_i  \geq \frac{c^{\log w} \cdot \delta^{\sum_{i=1}^t s
    2^i/2w}}{\left(w^3\log(1/\epsilon)\right)^{2\log w}} >
\frac{\delta}{\left(w\log(1/\epsilon)\right)^{C\log w}},
\end{align*}
for $C$ a sufficiently large constant.

Define the composition of two restrictions $\rho' \in \zro^L$ and $\rho'' \in \zro^{L(\rho')}$ in the natural way by $(\rho' \circ \rho'')_i = \rho''_i$ if $i \in L(\rho')$ and $(\rho' \circ \rho'')_i = \rho'_i$ otherwise. Then, by definition, we can view $\rho$ as a composition of independently chosen random restrictions $\rho_t \circ \rho_{t-1} \circ \cdots \circ \rho_1$, where $\rho_j \equiv \rho_{L_j,x^j}$ (with $x^j \in_u \zo^n$). Further, for any function $g$, $g_\rho \equiv  (((g_{\rho_1})_{\rho_2})_{\cdots})_{\rho_t}$. 

<<<<<<< .mine
Therefore, by iteratively applying the \clref{clm:slemma} $t$ times
with the random restrictions $\rho_1,\ldots,\rho_t$ and a union bound,
we get that with probability at least $1 - t(\delta^{s/4} +
\epsilon)$, there exists a lower approximating \dnf\ $f^\ell:\zo^L
\rgta \zo$ of width at most $w/2^{t+1}$ such that $f^\ell \leq f_\rho$
and $\bias(f_\rho) - \bias(f^\ell) < t \epsilon$. Similarly, by
iteratively applying the claim to the upper approximators given by the
claim, we get that with probability at least $1- 2 t (\delta^{s/4} +
\epsilon)$, $f_\rho$ has $(t\epsilon)$-sandwiching approximators that
are width-$s$ \dnf s.

Finally, the number of bits needed to sample $L$ is
\begin{align*}
r(n,s,\eps,\delta) = & \sum_{v=1}^t O\left(\frac{w}{2^v} \cdot \log(1/p(w/2^v, s)) + \log n\right) \\
& = O\left((\log n)(\log w)\right) + \sum_{v=1}^t \frac{w}{2^v}
\left(\frac{s 2^r}{2w} O\left(\log(1/\delta)\right) +
O\left(\log(w\log(1/\epsilon))\right) \right) \\ 
& = O\left((\log w) \cdot \left(\log n + s\log(1/\delta)\right) +
w\log\left(w\log(1/\epsilon)\right)\right).   
\end{align*}
The theorem now follows from applying the above argument to $\delta' =
\delta/2t$, $\epsilon' = \epsilon/t$ and noting that this only changes
the constant terms in the final bounds.  
\end{proof}


\section{Open Problems}
\label{sec:open}

A natural open question is to show optimal bounds for \dnf\
sparsification. We believe this question is interesting of its own right, even
without the sandwiching requirement. Formally, let $m(w,\eps)$ be
the smallest integer such that every width-$w$ \dnf \ formula can be
$\eps$-approximated by a width-$w$ \dnf \ with $m$
terms. \tref{th:mainstruct} shows that $m(w,\eps) \leq
(w\log(1/\eps)^{O(w)}$. Rocco Servedio \cite{Servedio} observed that the Majority
function on $2w$ variables (which is a width-$w$ \dnf) shows that
$m(w,\eps) \geq 4^{w - o(w)}$ for any constant $\eps$.
We are unaware of a better lower bound, and it is conceivable that the
right bound is exponential in $w$. We pose this as a conjecture: 

\begin{Conj}
\label{conj:gmr}

\noindent {\bf (Weaker Version)} There exists a function $c(\eps)$ such that
$$m(w,\eps) \leq O(c(\eps)^w).$$

\noindent {\bf (Stronger Version)} There exists a constant $c$ such that
$$m(w,\eps) \leq O(\log(1/\eps)^{cw}).$$
\end{Conj}

The weaker version, if true,  will imply that $\log(n)$ width \dnf s can be
$\eps$-approxmiated by $n^{O_\eps(1)}$ size \dnf s for any constant
$\eps$. Currently \tref{th:mainstruct}  gives the weaker bound of 
\begin{align*}
m(\log(n),\eps) \leq n^{O(\log\log(n)\log\log(1/\eps))}.
\end{align*}

The stronger version, if true, will strengthen Freidgut's theorem in the
context of \dnf s.

\paragraph{Mansour's Conjecture.}

Conjecture \ref{conj:gmr} is similar in spirit to Mansour's conjecture which
also asserts that \dnf \ formulas admit concise representations, but in
the Fourier domain. It also implies reductions between the conjecture
for small width \dnf s and small-size \dnf s.

We say that $f:\zo^n \rgta \zo$ has a $t$-sparse
$\eps$-approximation if there exists $p:\zo^n \rgta \R$ with at most
$t$ non-zero Fourier coefficients such that 
\begin{align*}
\E_{x \in \zo^n}\left[(f(x) - p(x))^2\right] \leq \eps.
\end{align*}

\begin{Conj}(Mansour's Conjecture for size)
\label{conj:mansour}
\cite{Mansour:94}

\noindent{\bf (Weaker version)} There exists a function $c(\eps)$ such that every
\dnf \ of size $m$ has an $m^{c(\eps)}$-sparse $\eps$-approximation. 

\noindent{\bf (Stronger version)} Every
\dnf \ of size $m$ has an $m^{O(\log(1/\eps))}$-sparse $\eps$-approximation. 
\end{Conj}

Mansour originally stated the stronger version of the conjecture, the
weaker version appears in \cite{Ryan-list}. The following analogue
of Mansour's conjecture for small width suggests itself. To our knowledge,
this conjecture has not appeared explicitly in the literature.
  
\begin{Conj}(Mansour's conjecture for width)
\label{conj:mansour-width}

\noindent{\bf (Weaker version)} There exists a function $c(\eps)$ such that every
\dnf \ of width $w$ has an $2^{c(\eps)w}$-sparse $\eps$-approximation. 

\noindent{\bf (Stronger version)} Every \dnf \ of width $w$ has an
$2^{O(w\log(1/\eps))}$-sparse $\eps$-approximation. 

\end{Conj}

The best known bounds for both size and width are due to Mansour, who
shows that every \dnf\ of width $w$ has an
$w^{O(w\log(1/\eps))}$-sparse $\eps$-approximation and then derives
a bound for size using $w = O(\log(m/\eps)$ \cite{Mansour:95}.

We feel that this width analogue of Mansour's Conjecture is natural;
indeed most results on \dnf s proceed by first tackling the width-$w$
case, and then translating it to \dnf s of size $m$ using $w \leq
\log(m/\eps)$ \cite{Hastad86,LMN:93,Mansour:95}.  
This substitution also shows that 
\begin{itemize}
\item The weaker version of Mansour's
Conjecture for width implies the weaker version of Mansour's Conjecture
for size. 
\item The stronger version of Mansour's Conjecture
for width implies the stronger version of Conjecture for size, as long
as $\eps \geq 1/\poly(m)$. 
\end{itemize}

Conjecture \ref{conj:gmr} implies the reverse equivalence.

\begin{Lem}

\begin{itemize}
\item Aussme the stronger version of Conjecture \ref{conj:gmr}. Then the stronger version of Mansour's Conjecture for size 
implies that every width $w$ \dnf\ formula has a
$2^{O(w\log(1/\eps)\log\log(1/\eps))}$-sparse
$\eps$-approximation.
\item Assume the weaker version of Conjecture \ref{conj:gmr}. Then the
  weaker version of Mansour's Conjecture for size implies the
  weaker version of Mansour's Conjecture for width. 
\end{itemize}
\end{Lem}

Note that if we replace Conjecture \ref{conj:gmr} with
\tref{th:mainstruct}, this does not improve on the bound from
\cite{Mansour:95}. So in this context, the improved
dependence on $w$ in Conjecture \ref{conj:gmr} is crucial.

\paragraph{Sparsification using the Greedy Algorithm.}
 
A natural approach to sparsifying a \dnf \ formula $f$ is to view it as a
set-covering problem, where we wish to cover $f^{-1}(1) \subseteq
\zo^n$ by width $w$ terms. One could use the greedy algorithm in the
hope that it constructs a sparse cover. It woule be interesting
to analyze its performance. In this direction, Jan Vondrak has pointed
out that one can use the analysis of greedy set cover to argue that if
there is a lower sandwiching \dnf \ formula of size $m_\ell(w,\eps)$
which is $\eps$-close to $f$, then greedy returns a $2\eps$
approximation of size at most $m_\ell(w,\eps)\ln(1/\eps)$ \cite{Vondrak}.  

\paragraph{Deterministic \dnf \ counting.}
The question of finding a deterministic polynomial time 
algorithm for approximate \dnf \ counting remains open. One approach towards this
goal would be to construct pseudorandom generators for \dnf s formulas
with seed-length $O(\log(n) + \log(m) + \log(1/\eps))$. Such
constructions are currently not known even for read-once \dnf s. A
recent result by the Trevisan, Vadhan and the authors gets a
seed-length of $\tilde{O}(\log(n) + \log(1/\eps))$ in the
read-once case \cite{GMRTV:12}.

\section*{Acknowledgements}
We thank Adam Klivans, Ryan O'Donnell, Rocco Servedio, Avi Wigderson and David
Zuckerman for valuable discussions. We thank Rocco for drawing our
attention to Friedgut's theorem in this context. 


\bibliographystyle{alpha}
\bibliography{references} 

\appendix
\section{Proofs from Section \ref{sec:sparsify}}

We first show that \lref{lm:sunflowercnf} is equivalent to Lemma
\ref{lm:sunflower} below from \cite{Rossman10}.
\begin{definition}[\cite{Rossman10}]
Let $\mc{F}$ be a family of sets over a universe $U$ and let $Y =
\cap_{T \in \mc{F}} S$. Call $\mc{F}$ a $\gamma$-sunflower if for a
random set $W \subseteq U$, with each element of $U$ present in $W$
independently with probability $1/2$, 
\begin{align*}
\pr[\exists T \in \mc{F}, (T\setminus Y) \cap W = \emptyset] \geq 1 -
\gamma.
\end{align*}
\end{definition}
\begin{lemma}[\cite{Rossman10}]\label{lm:sunflower}
Let $\mc{F}$ be a family of sets over a universe $U$ each of size at most $w$. If $|\mc{F}| > w! \cdot (2.47\log(1/\gamma))^w$, then $\mc{F}$ contains a $\gamma$-sunflower.
\end{lemma}

\begin{proof}[Proof of \lref{lm:sunflowercnf}]
As $f$ is unate, without loss of generality suppose that $f$ is
monotone. Let $U = [n]$ and $\mc{F} = \{T_i: 1\leq i \leq m\}$. By the
above lemma, there exists a $\gamma$-sunflower $\mc{F}' = \{T_{i_1},
\ldots, T_{i_s}\}$ for 
$$\gamma = \mu^{(m/w!)^{1/w}} \ \text{where} \ \mu = \frac{1}{2^{1/2.47}}.$$   

We claim that the lemma holds for the terms in $\mc{F}'$ and $Y = \cap_{j=1}^s T_{i_j}$. Let $x \in_u \zo^n$ and let $W = \{i: x_i = 0\}$. Then, each element of $U$ is present in $W$ independently with probability $1/2$. Therefore, as $\mc{F}'$ is a $\gamma$-sunflower 
\begin{align*}
  \pr_x[ \vee_{j=1}^c(T_{i_j} \setminus Y)  =1 ] =\pr_W[\exists T \in \mc{F}',\; (T \setminus Y) \cap W = \emptyset\,] \geq 1 - \gamma. 
\end{align*}
\end{proof}

We next show \newref[Fact]{fct:tedious}.
\begin{proof}[Proof of \fcref{fct:tedious}]
From the definition of $\gamma(\;)$ from \eref{eq:gamma}, it is easy
to check that $\gamma(j/2^w) \geq j^{1/w}/10 w$. We shall also use the
following inequality that follows from partial integration: for any
$\theta \geq k \geq  0$, 
\begin{equation}
  \label{eq:appendix1}
  \int_\theta^\infty x^k e^{-x} dx = \sum_{i = 0}^k \binom{k}{i} \cdot (i!) \cdot \left(\theta^{k-i} e^{-\theta}\right) \leq (k+1) \theta ^k \cdot e^{-\theta}.
\end{equation}
Therefore, for $\theta = W^{1/w}/10w$, 
\begin{align*}
  \sum_{j=W + 1}^{\infty} e^{-\gamma(j/2^w)} &\leq \sum_{j = W + 1}^\infty e^{-\left(j^{1/w}/10w\right)}\\
&\leq \int_{W}^\infty e^{-\left(x^{1/w}/10w\right)} \,dx\\
&= 10 w^2 \cdot (10 w)^{w-1} \cdot \int_{\theta}^\infty y^{w-1} \cdot e^{-y} dy \;\;\;\;\text{ (substituting $y \equiv x^{1/w}/10w$)}\\
&\leq 10 w^2 \cdot (10 w)^{w-1} \cdot w \cdot \theta^{w-1} e^{-\theta}  \;\;\;\;\;\;\;\text{ (by \eref{eq:appendix1})}\\
&\leq 10 w^3 \cdot W \cdot \exp(-10 w^2 \log(1/\epsilon))\\
&= \exp\left(\log(10 w^3) + w \log 2 + 3w \log w +  w\log(50 \log(1/\epsilon)) - 10 w^2 \log(1/\epsilon)\right)\\
&<\exp(-\log(1/\epsilon)) = \epsilon
\end{align*}
where the last inequality can be checked numerically for $w \geq 1$
and $\epsilon \leq 1/4$.  
\end{proof}

\section{Proofs from Section \ref{sec:dnf-count}}\label{app:LV}
In this section, we prove the two Lemmas from \cite{LubyV96} that are
used in our analysis. We restate them here for the reader's convenience.

\begin{Lem}(\lref{lem:lv1} Restated)
\label{lem:lv1-restated}
We have
\begin{align*}
\forall h \in \hh, \ \ \bias(f_h) \leq \bias(f),\\
\ex_{h \in \hh}[\bias(f_h)] \geq \bias(f_h)-  \epsilon.
\end{align*}
\end{Lem}
\begin{proof}
As $f_h$ is obtained by dropping terms in $f$, we have $f_h(x) \leq
f(x) \ \forall x \in \zo^n$, so $\bias(f_h) \leq \bias(f)$. 
This also implies that
\begin{align}
\label{eq:bias-fh}
\bias(f_h) = \frac{1}{2^n}\left(\sum_{x \in f^{-1}(1)}f_h(x)\right).
\end{align}
Taking expectation over $h$, we have
\begin{align}
\label{eq:bias-fh}
\E_{h \in \hh}[\bias(f_h)] = \frac{1}{2^n}\left(\sum_{x \in
  f^{-1}(1)}\E_{h \in \hh}[f_h(x)]\right).
\end{align}

Fix an $x \in f^{-1}(1)$ and a term $T_i$ of $f$ that it satisfies. 
If $T_i$ is included in $f_h$, which happens unless $T_i$ is bad for $h$,  then
$f_h(x) =1$. By \lref{lm:hashing} and a union bound,
\begin{align*}
\pr_{h \in \hh}[T_i \ \text{is bad for } h] \leq t\cdot 2^{-k}  \leq
\frac{\eps}{w}\cdot \frac{w}{k} \leq \eps.
\end{align*}
Hence we have
\begin{align*}
\E_{h \in \hh}[f_h(x)] \geq 1 -\eps.
\end{align*}
Plugging this into Equation \eqref{eq:bias-fh} gives
\begin{align*}
\E_{h \in \hh}[\bias(f_h)] \geq \frac{1}{2^n}\left(\sum_{x \in
  f^{-1}(1)}(1 - \eps)\right) = (1
-\eps)\frac{|f^{-1}(1)|}{2^n} = (1 -\eps)\bias(f).
\end{align*}
\end{proof}

\begin{Lem}(\lref{lem:lv2} restated)
\label{lem:lv2-restated}
We have 
\begin{align*}
|p_h - \bias(f_h)| \leq \eps.
\end{align*}
\end{Lem}
\begin{proof}
Let $\calD_0$ be the uniform distribution over $\zo^n$. For $j \in
[t]$, let $\calD_j$ be the distribution obtained from $\calD_{j-1}$ by
replacing the uniform distirbution on variables in bucket $B_j$ with
an independent copy of output of the generator $G$. Thus $\calD_t$ is
the output distribution of $G_h$.

We claim that for $j \in [t]$,
\begin{align}
\label{eq:hybrid-step}
\left| \pr_{x \in \calD_{j-1}}[f_h(x) = 1] - \pr_{x \in \calD_j}[f_h(x)
  = 1] \right| \leq \delta.
\end{align}
Since $\calD_{j-1}$ and $\calD_j$ differ only on the distribution over
bucket $B_j$, we first sample assignments for the other buckets.
The resulting formula on the variables in $B_j$ is a \dnf\ with width
at most $w'$.  Hence it is $\delta$-fooled by $G$, which gives
Equation \eqref{eq:hybrid-step}. 

We now have
\begin{align*}
|\bias(f_h) - p_h| & = \left|\E_{x \in \calD_0}[f_h(x)] - \E_{x \in
  \calD_t}[f_h(x)]\right| \\
& \leq \sum_{j=1}^t \left|\E_{x \in \calD_{j-1}}[f_h(x)] - \E_{x \in
  \calD_j}[f_h(x)]\right| \\
& \leq t\delta\\ 
& \leq \eps.
\end{align*}
\end{proof}


\end{document}